%% file: main.tex
\documentclass{article}

\input{macros}
\usepackage{centernot} 
\usepackage{mathtools}
\usepackage{coffeestains}

\def\thetahat{\hat{\theta}}

\def\A{\mathcal{A}}

\mathtoolsset{showonlyrefs}

\title{\vspace{-1cm}Conformal Risk Control for Non-Monotonic Losses}
\author{Anastasios N.\ Angelopoulos}
\date{Arena\\\texttt{anastasios@arena.ai}\\~\\ \today}

\begin{document}
\maketitle

\begin{abstract}
    Conformal risk control is an extension of conformal prediction for controlling risk functions beyond miscoverage.
    The original algorithm controls the expected value of a loss that is monotonic in a one-dimensional parameter.
    Here, we present risk control guarantees for generic algorithms applied to possibly non-monotonic losses with multidimensional parameters. 
    The guarantees depend on the stability of the algorithm---unstable algorithms have looser guarantees.
    We give applications of this technique to selective image classification, FDR and IOU control of tumor segmentations, and multigroup debiasing of recidivism predictions across overlapping race and sex groups using empirical risk minimization.
\end{abstract}

\section{Introduction}

Consider a sequence of random variables $D_{1:n+1} = ((X_1, Y_1), \ldots, (X_{n+1}, Y_{n+1}))$ representing feature-label pairs with an exchangeable joint distribution.
Let the first $n$ datapoints, $D_{1:n}$, be a known calibration set, and  last datapoint, $(X_{n+1},Y_{n+1})$, represent a test datapoint whose label is unknown.
Also define a bounded loss, $\ell(x,y;\theta) \in [0,1]$, which is a function of a datapoint $(x,y)$ and a parameter $\theta \in \R^d$.
Our goal is to select a parameter value $\thetahat$ using the calibration data $D_{1:n}$ to bound the expected loss on the test datapoint:
\begin{equation}
    \label{eq:intro-guarantee}
    \E\left[ \ell(X_{n+1}, Y_{n+1}; \thetahat) \right] \leq \alpha, 
\end{equation}
where $\alpha \in [0,1]$ is user-specified and the expectation is taken with respect to all $n+1$ datapoints.
We call~\eqref{eq:intro-guarantee} a risk control guarantee.

Conformal risk control, a generalization of conformal prediction~\cite{gammerman1998learning, vovk1999machine, vovk2005algorithmic,lei2013conformal,lei2018distribution} as developed in~\cite{angelopoulos2024conformal}, handles the case where $d=1$ and $\ell$ is monotonically nonincreasing in $\theta$.
Conformal risk control works by setting the parameter $\hat\theta$ to be
\begin{equation}
    \hat\theta = \inf\left\{ \theta : \frac{1}{n+1}\sum_{i=1}^{n}\ell(X_i,Y_i;\theta) \leq \alpha - \frac{1}{n+1} \right\},
\end{equation}
which is the smallest value of $\theta$ such that the empirical risk calculated on $n+1$ datapoints is certain to be below $\alpha$. This provides the guarantee in~\eqref{eq:intro-guarantee} for monotonic losses, but can fail arbitrarily badly for non-monotonic losses, as shown in Proposition 1 of~\cite{angelopoulos2024conformal}.

In this paper, we give risk control guarantees for non-monotonic losses.
The key insight is that the population risk of $\thetahat$ depends on the stability of the algorithm for choosing $\thetahat$---i.e., the change in risk when the test datapoint is added to or removed from the calibration dataset.
More formally, let $\A$ be an algorithm mapping datasets to choices of $\theta$ and let $D_{-i}$ denote $D_{1:n+1}$ with the $i$th element removed.
The parameter $\thetahat$ is the output of $\A(D_{1:n})$.
We say an algorithm $\A$ is $\beta$-stable with respect to a reference algorithm $\A^*$ and $\ell$ if
\begin{equation}
        \E\left[\frac{1}{n+1}\sum\limits_{i=1}^{n+1} \ell(X_{i}, Y_{i}; \A(D_{-i}))\right] 
        \leq \E\left[\frac{1}{n+1}\sum\limits_{i=1}^{n+1} \ell(X_{i}, Y_{i}; \A^*(D_{1:n+1}))\right] + \beta.
\end{equation}
The main theorem demonstrates that if an algorithm is stable with respect to a reference algorithm that controls the risk when run on the full data, then the original algorithm also controls the risk.
Here and throughout the paper, we assume the algorithms are symmetric (i.e. permutation-invariant).
\begin{theorem}
    \label{thm:main-intro}
    Assume $\A$ is symmetric and $\beta$-stable with respect to $\A^*$, that $D_{1:n+1}$ is exchangeable, and that
\begin{equation}
        \E\left[\ell(X_{n+1}, Y_{n+1}; \A^*(D_{1:n+1})) \right] \leq \alpha-\beta.
\end{equation}
    Then
    \begin{equation}
        \E\left[\ell(X_{n+1}, Y_{n+1}; \A(D_{1:n})) \right] \leq \alpha.
    \end{equation}
\end{theorem}
Theorem~\ref{thm:main-intro} gives us an actionable workflow for producing risk-control algorithms.
First, we try to identify a reference algorithm, $\A^*$, that controls the risk when applied on the full dataset of $n+1$ datapoints.
Then, we try to approximate that algorithm via $\A$, which runs only on $n$ datapoints, and prove that the difference in risks between these two algorithms is low.
The proof of the main result is below. 
\begin{proof}
By exchangeability and symmetry,
\[
\mathbb{E}\!\left[\frac{1}{n+1}\sum_{i=1}^{n+1}\ell(X_i,Y_i;\mathcal{A}^*(D_{1:n+1}))\right]
= \mathbb{E}\!\left[\ell(X_{n+1},Y_{n+1};\mathcal{A}^*(D_{1:n+1}))\right],
\]
and
\[
\mathbb{E}\!\left[\frac{1}{n+1}\sum_{i=1}^{n+1}\ell(X_i,Y_i;\A(D_{-i}))\right]
= \mathbb{E}\!\left[\ell(X_{n+1},Y_{n+1};\A(D_{1:n}))\right].
\]
By $\beta$-stability, $\mathbb{E}\!\left[\ell(X_{n+1},Y_{n+1};\A(D_{1:n}))\right] \leq \mathbb{E}\!\left[\ell(X_{n+1},Y_{n+1};\A^*(D_{1:n+1}))\right] + \beta$.
Combining this with the assumption that $\mathbb{E}[\ell(X_{n+1},Y_{n+1};\A^*(D_{1:n+1}))] \leq \alpha-\beta$ yields the result.
\end{proof}
Notice that Theorem~\ref{thm:main-intro} applies to $\theta$ in any space, and that boundedness of $\ell$ was never needed in the proof of Theorem~\ref{thm:main-intro}.
These concerns will come up later in determining stability bounds, building on the literature on algorithmic stability~\cite{kearns1997algorithmic, bousquet2002stability, kutin2002almost, mukherjee2006learning, shalev-shwartz2010learnability, yu2013stability, hardt2016train, yu2017three, feldman2018generalization, bousquet2020sharper, zrnic2023post}; see Chapter 13 of~\cite{shalev-shwartz2014understanding} for a review of this extensive field.
The notion of stability used herein is a form of leave-one-out stability with respect to generic algorithms $\cA$ and $\cA^*$; we will prove guarantees for a few different choices of algorithms depending on the setting.

\subsection{Conformal Risk Control is Stable for Monotonic Losses}

Here we will show that the conformal risk control algorithm is stable for monotonic losses, and thus, satisfies the guarantee in Theorem~\ref{thm:main-intro}.

\begin{proposition}
    \label{prop:crc-stable}
    Let $d=1$, $D$ be any dataset, $\ell$ be nonincreasing in its last argument, 
    \begin{equation}
        \A(D) = \inf\left\{ \theta : \frac{1}{|D|+1}\sum_{(x,y) \in D}\ell(x,y;\theta) \leq \alpha - \frac{1}{|D|+1} \right\}, 
    \end{equation}
    and
    \begin{equation}
        \A^*(D) = \inf\left\{ \theta : \frac{1}{|D|}\sum_{(x,y) \in D}\ell(x,y;\theta) \leq \alpha \right\}.
    \end{equation}
    Then $\A$ is $0$-stable with respect to $\A^*$.
\end{proposition}
\begin{proof}
    The boundedness of $\ell$ and definition of $\A$ imply that for all $i$,
    \begin{equation}
        \frac{1}{n+1}\sum_{j=1}^{n+1}\ell(X_j,Y_j;\A(D_{-i})) \leq \frac{1}{n+1}\sum_{j \in [n+1] \setminus i} \ell(X_j,Y_j;\A(D_{-i})) + \frac{1}{n+1} \leq \alpha.
    \end{equation}
    $\A(D_{-i})$ therefore satisfies the constraint in the definition of $\A^*(D_{1:n+1})$, and the monotonicity of the loss gives $\A(D_{-i}) \geq \A^*(D_{1:n+1})$. Therefore, $\ell(X_{i}, Y_{i}; \A(D_{-i})) \leq \ell(X_{i}, Y_{i}; \A^*(D_{1:n+1}))$, and
    \begin{equation}
        \E\left[\frac{1}{n+1}\sum\limits_{i=1}^{n+1} \Big(\ell(X_{i}, Y_{i}; \A(D_{-i})) - \ell(X_{i}, Y_{i}; \A^*(D_{1:n+1}))\Big)\right] \leq 0.
    \end{equation}
    Rearranging terms gives $0$-stability.
\end{proof}

Combining Proposition~\ref{prop:crc-stable} with Theorem~\ref{thm:main-intro} exactly recovers the classical conformal risk control guarantee.
\begin{corollary}
    The conformal risk control algorithm $\A$ defined in Proposition~\ref{prop:crc-stable} satisfies
    \begin{equation}
        \E\left[ \ell(X_{n+1},Y_{n+1};\A(D_{1:n})) \right] \leq \alpha.
    \end{equation}
\end{corollary}
\begin{proof}
    For monotonic losses, Proposition~\ref{prop:crc-stable} shows that the conformal risk control algorithm $\mathcal{A}$ is $0$-stable. Applying Theorem~\ref{thm:main-intro} with $\beta=0$ recovers the guarantee.
\end{proof}

The key takeaway is that the proof of conformal risk control's validity has two parts, which can be decoupled: (1) proving monotonicity implies stability, and (2) proving that stability implies risk control.
Theorem~\ref{thm:main-intro} tells us that (2) holds for general algorithms.
The remainder of the paper will be dedicated to replacing step (1) with other stable algorithms, in the context of non-monotonic losses.
Looking forward, Section~\ref{sec:methods} gives stability bounds for classes of non-monotonic losses and shows how to estimate these stability bounds for use in Theorem~\ref{thm:main-intro}. Section~\ref{sec:experiments} shows experimental validation of these methods.

\section{Methods}
\label{sec:methods}

\subsection{Setup and Notation}
\label{sec:setup}

Let $\cX$ be the feature space, $\cY$ be the label space, and $\cZ = \cX \times \cY$.
Datasets are ordered vectors of feature-label pairs, denoted $D \in \cZ^{*}$, where the notation $\cZ^*$ refers to the space of any-length sequences with elements in $\cZ$.
A dataset $D$ can have repeated entries; the size of the dataset is denoted as $|D|$.
To refer to the dataset of random variables used throughout this paper, we will use the notation $D_{1:k} = ((X_1,Y_1), \ldots, (X_{k},Y_{k})) = (Z_1, \ldots, Z_{k})$.
Furthermore, we use $D_{-i}$ to denote the dataset with the $i$th element removed.

We will let $\hat{R}_{D}(\theta) = \tfrac{1}{|D|}\sum_{(x,y) \in D}\ell(x,y;\theta)$ denote the empirical risk, with the shorthand $\hat{R}_{1:k}(\theta)=\hat{R}_{D_{1:k}}(\theta)$ and $\hat{R}_{-i}(\theta)=\hat{R}_{D_{-i}}(\theta)$. 
We also denote the population risk as $R(\theta)=\E[\ell(X_{n+1},Y_{n+1};\theta)]$.
Also, we will let $\thetahat_k = \A(D_{1:k})$ and $\thetahat_{-i} = \A(D_{-i})$ for all $k$ and $i$.
When $\A$ is $\beta$-stable with respect to $\A$ and $\ell$ is clear from context, we will simply say it is $\beta$-stable.
We will also use the notation $[k] = \{ 1, \ldots, k \}$ for any positive integer $k$.

\subsection{Stable Algorithms}

This section describes three stable classes of risk control algorithms.
We begin with the case of general bounded losses, which we handle via discretization.
Next, we give a stronger result for continuous, Lipschitz losses with a strong crossing point at $\alpha$.
Finally, we perform a narrow case study with explicit constants: selective classification, otherwise known as classification with abstention.
By abstaining from prediction on the hardest examples, selective classification allows us to improve a model's accuracy on the event that it issues a prediction.
Selective classification is one of the most important applications of non-monotonic conformal risk control due to its ubiquitous practical utility in safe automation.
Throughout this section, we will set $d=1$.

\subsubsection{General Bounded Losses}
\label{sec:general-bounded}
To handle general losses bounded in $[0,1]$, we will use algorithms that discretize $\Theta$ to impose stability.
First, let $\A(D) = \inf\{\theta \in \Theta_m = \{0, \tfrac{1}{m}, \tfrac{2}{m}, \ldots, 1\} : \hat{R}_D(\theta) \leq \alpha\}$ for $\Theta = [0,1]$ and any positive integer $m$.
This is simply the discretized version of the previously studied algorithms.
We will assume nothing about the loss other than the existence of a safe solution, $\ell(z;1)=0$ for all $z \in \cX\times\cY$.
\begin{proposition}
\label{prop:general-stability}
Let $D_{1:n+1}$ be sampled i.i.d., $\Theta_m=\{0,\tfrac1m,\ldots,1\}$, $\ell(\cdot;1)=0$, and define
\begin{equation}
\A(D)\ :=\ \inf\bigl\{\theta\in\Theta_m:\ \hat R_D(\theta)\le \alpha\bigr\}.
\end{equation}
Then
\begin{equation}
\E\bigl[R(\thetahat_{n})\bigr]\ \leq\ \alpha+\frac{1}{2\sqrt{n}}\left(\sqrt{-W_{-1} \left( -\frac{1}{4n(m+1)^2}\right)}+\sqrt{-\frac{1}{W_{-1} \left( -\frac{1}{4n(m+1)^2}\right)}}\right) = \alpha + \tilde{\mathcal{O}}\left(\frac{1}{\sqrt{n}}\right),
\end{equation}
where $W_{-1}$ is the $-1$st branch of the Lambert W function and $\tilde{\mathcal{O}}$ denotes the growth rate excluding logarithmic factors.
\end{proposition}
This proposition tells us that we can achieve risk control for general losses with the discretized algorithm up to a dominating factor of $\tilde{\mathcal{O}}\left(\tfrac{1}{\sqrt{n}}\right)$, since $W_{-1}$ is of logarithmic order.

\subsubsection{Continuous, Lipschitz Losses}
Next, we study the same root-finding algorithm, $\A(D) = \inf\{\theta \in \R : \hat{R}_D(\theta) \leq \alpha\}$, on losses that satisfy regularity conditions.
If the root is leave-one-out stable, meaning that removing a single datapoint does not drastically change the location of the root, then for smooth losses, the algorithm will control the risk.

\begin{proposition}
    \label{prop:stability-smooth}
    Let $d=1$, $\A(D) = \inf\{\theta \in \R : \hat{R}_D(\theta) \leq \alpha\}$, and $D_{1:n+1}$ be exchangeable.
    Let $\ell$ be continuous and $L$-Lipschitz in $\theta$, and assume that $\thetahat_{n+1}$ is almost surely finite.
    Furthermore, assume there exist $m,r>0$ such that
    \begin{multline}
        \label{eq:stable_root}
        \hat{R}_{1:n+1}(\theta) \leq \alpha - m(\theta - \thetahat_{n+1}) \; \; \forall \theta \in [\thetahat_{n+1}, \thetahat_{n+1} + r ], \quad \hat{R}_{1:n+1}(\theta) \geq \alpha + m(\thetahat_{n+1} - \theta) \; \; \forall \theta \in [\thetahat_{n+1} - r, \thetahat_{n+1} ],\\
        \text{ and } \quad \hat{R}_{1:n+1}(\theta) \geq \alpha + mr \; \; \forall \theta \in [-\infty, \thetahat_{n+1}-r)      
    \end{multline}
    almost surely. Then if $\tfrac{1}{n+1} < mr$, $\A$ is $\tfrac{L}{m(n+1)}$-stable with respect to $\A^*\equiv\A$.
\end{proposition}
This result says that if the empirical risk satisfies:
\begin{enumerate}
    \item local linearity with sufficient slope at its leftmost crossing of $\alpha$, and
    \item does not come too close to $\alpha$ after the leftmost crossing,
\end{enumerate}
then the risk is controlled.
Intuitively, the risk cannot touch $\alpha$ in two highly disjoint regions, otherwise the selected parameter may be unstable.
\begin{proof}
    Set $\epsilon:=1/(n+1)$. 
    First, observe the following leave-one-out bound:
    \begin{equation}
        \label{eq:loo-bound}
        \sup_{\theta\in\Theta}\,|\hat{R}_{-i}(\theta)-\hat{R}_{1:n+1}(\theta)|\ \leq \epsilon.
    \end{equation}
    Applying the assumptions in~\eqref{eq:stable_root} gives
    \begin{equation}
        \hat{R}_{1:n+1}\left(\hat\theta_{n+1}-\delta\right) \geq \alpha + m\delta \quad \text{ and } \quad
        \hat{R}_{1:n+1}\left(\hat\theta_{n+1}+\delta\right) \leq \alpha-m\delta \qquad \forall\delta\in(0,r].
    \end{equation}
    Combining the above display applied with $\delta=\epsilon/m$ and the leave-one-out bound gives
    \begin{equation}
        \hat{R}_{-i}\left(\hat\theta_{n+1}-\delta\right) \geq \alpha + m\tfrac{\epsilon}{m} - \epsilon = \alpha, \quad \text{ and } \quad \hat{R}_{-i}\left(\hat\theta_{n+1}+\delta\right) \leq \alpha-m\tfrac{\epsilon}{m}+\epsilon=\alpha,
    \end{equation}
    implying that $\alpha \in \left[ \hat{R}_{-i}\!\left(\hat\theta_{n+1}\pm \tfrac{\epsilon}{m}\right)\right]$.
    Therefore, by continuity, $\hat{R}_{-i}(\theta)=\alpha$ has a solution in
    \[
    I:=\Big[\hat\theta_{n+1}-\tfrac{\epsilon}{m},\ \hat\theta_{n+1}+\tfrac{\epsilon}{m}\Big].
    \]
    Next, for any $\theta\in [\hat\theta_{n+1}-r, \hat\theta_{n+1}-\epsilon/m)$~\eqref{eq:stable_root} combines with~\eqref{eq:loo-bound} to yield
    \begin{equation}
        \hat{R}_{-i}(\theta) \geq \alpha + m(\hat\theta_{n+1}-\theta)-\epsilon\ > \alpha,
    \end{equation}
    so there is no root of $\hat{R}_{-i}(\theta)=\alpha$ in that interval. For $\theta < \hat\theta_{n+1} - r$, the second assumption in~\eqref{eq:stable_root} combines with~\eqref{eq:loo-bound} to give
    \begin{equation}
        \hat{R}_{-i}(\theta) \geq \alpha + mr - \epsilon > \alpha,
    \end{equation}
    so there is no root to the left of $\hat\theta_{n+1}-r$ either. Hence the leftmost root of $\hat{R}_{-i}(\theta)-\alpha$ lies in $I$, and therefore
    \begin{equation}
        \big|\thetahat_{-i}-\thetahat_{n+1}\big|\ \le\ \frac{\epsilon}{m}\ =\ \frac{1}{m(n+1)}.
    \end{equation}
    Finally, since $\ell(\cdot;\theta)$ is $L$-Lipschitz in $\theta$,
    \[
    \ell\!\left(Z_i;\thetahat_{-i}\right)
    \le \ell\!\left(Z_i;\thetahat_{n+1}\right)
    + L\,\big|\thetahat_{-i}-\thetahat_{n+1}\big|
    \le \ell\!\left(Z_i;\thetahat_{n+1}\right)+\frac{L}{m(n+1)}.
    \]
    Averaging over $i=1,\dots,n+1$ and taking expectations on both sides completes the proof.
\end{proof}

Combining the previous result with Theorem~\ref{thm:main-intro} gives a risk-control result.
\begin{corollary}
    Under the same conditions as Proposition~\ref{prop:stability-smooth}, $\A$ satisfies
    \begin{equation}
        \E\left[ R(\thetahat_{n}) \right] \leq \alpha + \frac{L}{m(n+1)}.
    \end{equation}
\end{corollary}

\subsubsection{Selective Classification}

Finally, we give stability guarantees for selective classification.
Let $\hat{Y}_i$ represent a prediction of $Y_i$ and $\hat{P}_i \in [0,1]$ represent a confidence in this prediction. Both $\hat{Y}_i = \hat{y}(X_i)$ and $\hat{P}_i = \hat{p}(X_i)$ are outputs of deterministic models run on $X_i$. We seek a confidence threshold $\thetahat$ satisfying
\begin{equation}
    \P(\hat{Y}_{n+1} \neq Y_{n+1} \mid \hat{P}_{n+1} > \thetahat) \leq \alpha.
\end{equation}
As described in~\cite{angelopoulos2025learn}, this is equivalent to asking that
\begin{equation}
    \E\left[ \ind{\hat{Y}_{n+1} \neq Y_{n+1} \text{ and } \hat{P}_{n+1} > \thetahat} - \alpha \ind{\hat{P}_{n+1} > \thetahat} + \alpha \right] = \E\left[ \ell(Z_{n+1};\thetahat) \right] \leq \alpha,
\end{equation}
where $\ell(x,y;\theta) = \ind{\hat{y}(x) \neq y \text{ and } \hat{p}(x) > \theta} - \alpha \ind{\hat{p}(x) > \theta} + \alpha$.
Notably, $\ell$ is a non-monotonic loss.
We will study the below algorithm $\A$ for selective classification,
\begin{equation}
    \label{eq:leftmost-root}
    \A(D) = \inf\left\{ \theta : \hat{R}_{D}(\theta) \leq \alpha \right\},
\end{equation}
which simply chooses the smallest value of $\theta$ that controls the risk.

First, we characterize the stability of this selective classification algorithm.
Let the $\hat P_i$ be distinct almost surely.
Let the vector $V$ denote the inverse sorting permutation of $\hat P$, so that $\hat{P}_i = \hat{P}_{(V_i)}$, where the notation $\hat{P}_{(j)}$ represents the $j$th smallest value of $\hat{P}$.
Further, let $V^==\{V_i : \hat{Y}_{i} = Y_{i} \}$, $V^{\neq} =\{V_i : \hat{Y}_{i} \neq Y_{i} \}$.
The core idea is that each loss is piecewise constant, with one component corresponding to when the algorithm is abstaining and another when it issues a prediction.
Thus it only has one change point, exactly at the location of $\hat{P}_i$:
\begin{align}
    \label{eq:selective-eq-neq}
    \ell(Z_i ; \theta) = 1 - (1-\alpha)\ind{\theta \geq \hat{P}_i} \qquad & \forall i \in V^{\neq}, \text{ and }\\
    \ell(Z_i ; \theta) = \alpha\ind{\theta \geq \hat{P}_i} \qquad & \forall i \in V^=.
\end{align}
Because the losses are right-continuous and piecewise constant with jumps at $\hat{P}$, we are guaranteed that if $D$ is a subdataset of $D_{1:n+1}$, then $\A(D)=\hat{P}_{(j)}$ for some $j$. 
Along these lines, define $\hat{\jmath}_{n+1}$ and $\hat{\jmath}_{-i}$ as
\begin{equation}
    \hat{P}_{(\hat{\jmath}_{n+1})} = \thetahat_{n+1} \text{ and } \hat{P}_{(\hat{\jmath}_{-i})} = \thetahat_{-i}.
\end{equation}
We will use this index-space characterization to bound the stability.

\begin{proposition}
    \label{prop:stability-selective}
    The selective classification algorithm in~\eqref{eq:leftmost-root} is $\beta$-stable with
    \begin{equation}
        \beta = \frac{2\max\{\alpha, 1-\alpha\}\E[K]}{n+1}, \text{ where } K = \max_i|\hat{\jmath}_{-i} - \hat{\jmath}_{n+1}|.
    \end{equation}
\end{proposition}
\begin{proof}
    By~\eqref{eq:selective-eq-neq}, it suffices to prove that, almost surely,
    \begin{equation}
        \sum\limits_{i=1}^{n+1} \Big(\ell(X_{i}, Y_{i}; \A(D_{-i})) - \ell(X_{i}, Y_{i}; \A^*(D_{1:n+1}))\Big) = \sum_{i\in [n+1]} w_i\Delta_i
        \leq 2\max\{\alpha, 1-\alpha\}K,
    \end{equation}
    where $\Delta_i = \ind{\A(D_{-i}) \geq \hat P_i} - \ind{\A^*(D_{1:n+1}) \geq \hat P_i}$ and $w_i = \alpha$ for $i \in V^{=}$ and $-(1-\alpha)$ otherwise.
    We can begin by rewriting the $\Delta_i$ terms as 
    \begin{align}
        \Delta_i &= \ind{\A(D_{-i}) \geq \hat{P}_i > \A(D_{1:n+1})} - \ind{\A(D_{-i}) < \hat{P}_i \leq \A(D_{1:n+1})} \\
        &= \ind{\hat{\jmath}_{-i} \leq V_i < \hat{\jmath}_{n+1}} - \ind{\hat{\jmath}_{-i} > V_i \geq \hat{\jmath}_{n+1}}.
    \end{align}
    Note that the two events inside the indicators are mutually exclusive, and correspond to the ``crossing'' of $\A(D_{-i})$ over $\hat{P}_i$ in the downward or upward directions.
    Therefore,
    \begin{equation}
        \label{eq:pf-selective-cauchy-schwarz}
        \sum_{i\in [n+1]} w_i\Delta_i \leq \max\{\alpha, 1-\alpha\} \sum_{i\in [n+1]}|\Delta_i| \leq \max\{\alpha, 1-\alpha\}\left|\{i : \hat{\jmath}_{-i} \leq V_i < \hat{\jmath}_{n+1} \text{ or } \hat{\jmath}_{-i} > V_i \geq \hat{\jmath}_{n+1}\} \right|.
    \end{equation}
    Define the index sets
\[
S_- \;:=\; \bigl\{i\in[n{+}1]:\ \hat{\jmath}_{-i}\le V_i<\hat{\jmath}_{n+1}\bigr\}
\quad\text{and}\quad
S_+ \;:=\; \bigl\{i\in[n{+}1]:\ \hat{\jmath}_{-i}> V_i\ge \hat{\jmath}_{n+1}\bigr\}.
\]
Then the set of indices that contribute nonzero terms to $\sum_{i\in[n+1]}|\Delta_i|$ is $S:=S_-\cup S_+$, and $S_-\cap S_+=\varnothing$ by definition.

By the definition of $K=\max_i|\hat{\jmath}_{-i}-\hat{\jmath}_{n+1}|$, for every $i\in S_-$ we have
\[
1 \le \hat{\jmath}_{n+1}-V_i \le \hat{\jmath}_{n+1}-\hat{\jmath}_{-i}\le |\hat{\jmath}_{-i}-\hat{\jmath}_{n+1}|\le K,
\]
hence
\[
V_i\in\{\hat{\jmath}_{n+1}-K,\ldots,\hat{\jmath}_{n+1}-1\}\cap\{1,\ldots,n{+}1\}.
\]
Similarly, for every $i\in S_+$ we have $\hat{\jmath}_{n+1}\le V_i\le \hat{\jmath}_{-i}-1\le \hat{\jmath}_{n+1}+K-1$, i.e.
\[
V_i\in\{\hat{\jmath}_{n+1},\ldots,\hat{\jmath}_{n+1}+K-1\}\cap\{1,\ldots,n{+}1\}.
\]
Because the $V_i$’s are distinct (the map $i\mapsto V_i$ is a permutation of $\{1,\ldots,n{+}1\}$), at most one index $i$ can correspond to each admissible value of $V_i$. Therefore,
\[
|S_-|\le K\qquad\text{and}\qquad |S_+|\le K,
\]
which, since $|\Delta_i| \leq 1$ for all $i$, yields
\[
\sum_{i\in[n+1]}|\Delta_i| \leq |S| \leq |S_-|+|S_+| \leq 2K.
\]
Combining this with~\eqref{eq:pf-selective-cauchy-schwarz} yields
\[
\sum_{i\in[n+1]} w_i\Delta_i \;\le\; \max\{\alpha, 1-\alpha\}\,|S| \;\le\; 2\max\{\alpha, 1-\alpha\}K,
\]
as claimed.
\end{proof}

The above result gives a simple, distribution-free characterization of the stability of the selective classification algorithm.
It is then relatively easy to estimate $\E[K]$ with the bootstrap---i.e., drawing samples from the empirical distribution, computing $K$ on these samples, and taking the average.
When combined with Theorem~\ref{thm:main-intro}, this stability result implies a selective accuracy bound.
\begin{corollary}
    Let $D_{1:n+1}$ be exchangeable. Then the selective classification algorithm in~\eqref{eq:leftmost-root} satisfies
    \begin{equation}
        \P(\hat{Y}_{n+1} = Y_{n+1} \mid \hat{P}_{n+1} > \thetahat_n) \geq 1-\alpha-\frac{2\max\{\alpha, 1-\alpha\}\E[K]}{n+1}.
    \end{equation}
\end{corollary}

However, the question remains as to a more explicit characterization of $\E[K]$.
The following proposition gives an answer, saying that the expected value of $K$ is bounded by the number of times the algorithm's running error rate crosses a thin band around $\alpha$, as visualized in Figure~\ref{fig:shrinking-interval}.
\begin{figure}[ht]
    \centering
    \includegraphics[width=0.8\linewidth]{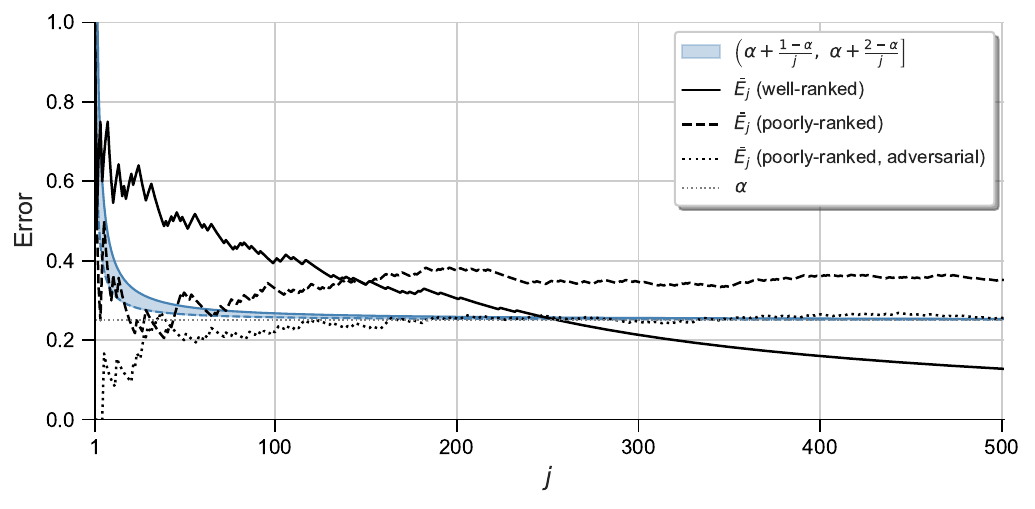}
    \caption{Cumulative average error $\bar{E}_j$ for three scenarios alongside the shrinking interval $(\alpha + (1-\alpha)/j,\; \alpha + (2-\alpha)/j]$ from Proposition~\ref{prop:ub-EK}, with $\alpha = 0.25$ and $n = 500$.
    In each case $E_i \sim \mathrm{Bernoulli}(p_i)$ independently.
    In the well-ranked case, $p_i = \max(0,\; 2\alpha(1 - (i-1)/m))$, $m = (n+1)/2$, a linear decay from $2\alpha$ to $0$ at the midpoint, so that $\bar{E}_j$ crosses below $\alpha$ in the second half and eventually exits the interval.
    In the poorly ranked case, $p_i = 0.35$ for all $i$; the constant rate exceeds $\alpha$, and $\bar{E}_j$ remains above the interval.
    In the poorly ranked adversarial case, $p_i = \alpha$ for all $i$; the average concentrates near $\alpha$ but stays below the interval as it shrinks toward $\alpha$ from above.
    The blue band is the interval; the gray dotted line marks $\alpha$. $K$ is the number of $j$ for which $\bar{E}_j$ lands inside the interval.}
    \label{fig:shrinking-interval}
\end{figure}
\begin{proposition}
    \label{prop:ub-EK}
    In the same setting as Proposition~\ref{prop:stability-selective}, define $\bar{E}_j = \frac{1}{j}\sum_{i \leq j}E_i$, where
    \begin{equation}
        E = \left(\ind{\hat{Y}_{V_1} \neq Y_{V_1}}, \ldots, \ind{\hat{Y}_{V_{n+1}} \neq Y_{V_{n+1}}}\right)
    \end{equation}
    Then
    \begin{equation}
        \E[K] \leq \sum_{j=1}^{n+1} \left[\P\left(\bar{E}_j \in \bigg(\alpha + \frac{1-\alpha}{j}, \alpha + \frac{2-\alpha}{j}\bigg] \right)\right].
    \end{equation}
\end{proposition}
Because in practical settings we use models that tend to have a reasonable ranking, we therefore expect $K$ to be small in practice, as in Figure~\ref{fig:shrinking-interval}.
This intuition can be further formalized by applying a concentration bound like Hoeffding's inequality on $\bar{E}_j$ to bound the probability it lies in $\big(\alpha + \tfrac{1-\alpha}{j}, \alpha + \tfrac{2-\alpha}{j}\big]$ conditionally on the vector $\hat{P}$. The validity of such an argument hinges on the fact that the coordinates of $E$ are independent Bernoulli random variables with different means after conditioning on $\hat{P}$; the result is not too interpretable, so we omit it here.
\begin{proof}
    We keep the same notation as in the proof of Proposition~\ref{prop:stability-selective},
    and notice that 
    \begin{equation}
        K = \max_i |\A'(D_{-i})-\A'(D_{1:n+1})|,
    \end{equation}
    with
    \begin{align}
        \hat{\jmath}_{n+1} = \A'(D_{1:n+1}) &= \max
        \left\{j : \sum_{v \in V^{=}}\alpha \ind{j \leq v} + \sum_{v \in V^{\neq}}(1-(1-\alpha)\ind{j \leq v}) \leq (n+1)\alpha \right\} \\
        &= \max\big\{j \in \{0, \ldots, n+1\} :\;1+T_j\ge 0\big\},
    \end{align}
    where $T_j = \sum_{i \leq j}w_i = -\sum_{i\leq j}E_i+j\alpha$.
    The leave-one-out index $\hat{\jmath}_{-i}=\mathcal{A}'(D_{-i})$ is defined similarly, with $V^=_{-i}$ in place of $V^=$, $V^{\neq}_{-i}$ in place of $V^{\neq}$, and $n$ in place of $n+1$.
    Therefore, $K$ depends on the values of $\hat{P}_1, \ldots, \hat{P}_{n+1}$, only through the error vector $E$ (which itself only depends on their order).
    If we delete the datapoint at index $V_i$, then for all $j\geq V_i$ the new partial sums shift by a constant:
    \begin{equation}
        T^{-i}_{j}=T_{j+1}-\alpha\quad\text{if }E_i=0,
        \qquad
        T^{-i}_{j}=T_{j+1}+(1-\alpha)\quad\text{if }E_i=1.
    \end{equation}
    Hence removing a correct point ($E_i=0$) can only invalidate those $j \geq V_i$ for which $1+T_{j+1}\in[0,\alpha)$, and removing an error ($E_i=1$) can only create feasibility for those $j \geq V_i$ with $1+T_{j+1}\in[-(1-\alpha),0)$. Consequently,
    \begin{equation}
        K \leq \bigl|\{j \in [n] :\ T_j\in[-1,\alpha-1)\}\bigr| + \bigl|\{j \in [n] :\ T_j \in [\alpha-2, -1)\}\bigr| = \bigl|\{j \in [n] :\ T_j\in[\alpha-2,\alpha-1)\}\bigr|
    \end{equation}

    We can bound this quantity in expectation as
    \begin{align}
        \label{eq:EK-ub-1}
        &\E[K] \leq \E\big[\bigl|\{j \in [n] :\ T_j\in[\alpha-2,\alpha-1)\}\bigr|\big] \\
        = &\E\left[\sum\limits_{j=1}^{n+1} \ind{\alpha-2 \leq T_j < \alpha-1}\right] = \sum\limits_{j=1}^{n+1} \P\left(T_j \in [\alpha-2, \alpha-1) \right)=\sum\limits_{j=1}^{n+1}\P\left(\bar{E}_j \in \bigg(\alpha + \frac{1-\alpha}{j}, \alpha + \frac{2-\alpha}{j}\bigg] \right).
    \end{align}
\end{proof}

\subsection{Guarantees for Empirical Risk Minimization}
So far, we have studied root-finding algorithms that select the most extreme point satisfying an inequality.
Here, we will instead study regularized empirical risk minimization (ERM) algorithms: 
\begin{equation}
    \A(D) = \argmin_{\theta \in \R^d} \hat{R}_D(\theta) + \frac{\lambda}{2} \|\theta\|_2^2,
\end{equation}
for some regularization level $\lambda \geq 0$. We will let $\ell : \R^d \to \R$ be potentially unbounded, and $d \geq 1$.

We will show two types of guarantees with respect to ERM.
First, we prove guarantees on the generalization gap: the gap between the risk of the empirical risk minimizer with respect to an oracle quantity, such as the full-data or population risk minimizers.
Second, we prove stability bounds on the first-order optimality condition; i.e., how close the population gradient is to zero.

\subsubsection{Risk Control Guarantees on the Loss}
First, we show the stability of ERM on the scale of the loss.
\begin{proposition}
\label{prop:erm-stable-loss}
Let $\Theta=\R^d$, $\ell$ be convex in $\theta$, and assume there exists $\rho:\cZ\to[0,\infty)$ such that for all $z\in\cZ$ and $\theta,\theta'\in\R^d$,
\begin{equation}
\label{eq:lipschitz-rho-simple}
|\ell(z;\theta)-\ell(z;\theta')|\ \le\ \rho(z)\,\|\theta-\theta'\|_2.
\end{equation}
Then $\A$ is $\beta$-stable with
\begin{equation}
\label{eq:beta-simple}
\beta \leq \frac{2\,\E[\rho(Z_{n+1})^2]}{\lambda(n+1)}.
\end{equation}
\end{proposition}
This result can be used in conjunction with Theorem~\ref{thm:main-intro} to give a generalization guarantee on ERM, as below.
\begin{corollary}
    In the setting of Proposition~\ref{prop:erm-stable-loss},
    \begin{equation}
        \E\left[ \ell(X_{n+1}, Y_{n+1}; \thetahat_n) + \tfrac{\lambda}{2}\|\thetahat_n\|_2^2 \right] \leq R^* + \frac{2\,\E[\rho(Z_{n+1})^2]}{\lambda(n+1)},
    \end{equation}
    where $R^* = \min_{\theta \in \Theta} \E[\ell(X_{n+1},Y_{n+1};\theta)+ \tfrac{\lambda}{2}\|\theta\|_2^2]$.
\end{corollary}
These results represent minor variations on the canonical results in Section 5 of~\cite{bousquet2002stability}.

\subsubsection{Risk Control Guarantees on the Gradient}

Next, we address stability of the expected gradient.
Building towards this, we will need a multivariate notion of stability to handle gradients when $d>1$: an algorithm $\cA$ is $\beta$-stable with respect to $\A^*$ and $g$ if
\begin{equation}
        \E\left[\frac{1}{n+1}\sum\limits_{i=1}^{n+1} g(X_{i}, Y_{i}; \A(D_{-i}))\right] 
        \preceq \E\left[\frac{1}{n+1}\sum\limits_{i=1}^{n+1} g(X_{i}, Y_{i}; \A^*(D_{1:n+1}))\right] + \beta,
\end{equation}
where $\beta \in \R^d$, and $\preceq$ represents the standard partial ordering on vectors (i.e., the inequality holds component-wise).
Theorem~\ref{thm:main-intro} also has an extension to $d>1$, which we present below.
\begin{theorem}
    \label{thm:main-high-d}
    Assume $\A$ is symmetric and $\beta$-stable with respect to $\A^*$ and $g$, that $D_{1:n+1}$ is exchangeable, and that
\begin{equation}
        \E\left[g(X_{n+1}, Y_{n+1}; \A^*(D_{1:n+1})) \right] \preceq \alpha-\beta.
\end{equation}
    Then
    \begin{equation}
        \E\left[g(X_{n+1}, Y_{n+1}; \A(D_{1:n})) \right] \preceq \alpha.
    \end{equation}
\end{theorem}
The proof of this theorem is identical to that of Theorem~\ref{thm:main-intro}, except with $\preceq$ in place of $\leq$.
Now we proceed with the stability result, stated below for differentiable convex losses. The differentiability is not needed, and is only used to simplify the statement and proof of the proposition.
\begin{proposition}
\label{prop:erm-grad-stability}
Assume $\ell$ is a convex, differentiable function of $\theta$ and that there exists a measurable function $\rho : \cZ \to [0,\infty)$ such that for all $\theta,\theta' \in \R^d$ and almost all $z \in \cZ$,
\begin{equation}
\|\nabla\ell(z;\theta) - \nabla\ell(z;\theta')\|_2 \leq \rho(z)\|\theta - \theta'\|_2.
\end{equation}
Assume also that $\hat{R}_{1:n+1}$ is $\mu$-strongly convex for some $\mu \in [0,\infty)$ (noting that $\mu=0$ is also allowable, so it need not be strongly convex).
Then the regularized ERM algorithm $\cA$ is $\beta$-stable with respect to $\nabla\ell$, where
\begin{equation}
\beta = \frac{\E[\rho(Z_{n+1})\|\nabla\ell(Z_{n+1}; \thetahat_n)\|_2] + \E\left[\rho(Z_{n+1})\|\tfrac{1}{n}\sum_{j=1}^n\nabla\ell(Z_j;\thetahat_n)\|_2\right]}{(\mu + \lambda)(n+1)} \mathbf{1}_d.
\end{equation}
\end{proposition}
Eliding the Lipschitz constant $\rho(Z_{n+1})$, this result tells us that ERM is stable with
\begin{equation}
    \beta = \frac{\E[\text{magnitude of test gradient}] + \E[\text{magnitude of training gradient}]}{(\mu + \lambda)(n+1)}.
\end{equation}
Under normal circumstances, both terms in the numerator will be constant-order, and the latter term will be near-zero (in fact, when $\lambda=0$, it is exactly zero).
Proposition~\ref{prop:erm-grad-stability} also holds for nondifferentiable losses, replacing every gradient in $\beta$ with a supremum over all subgradients in the subdifferential at $\thetahat_{n}$; the proof of this more general result is morally the same, so we omit it.
The assumption that $\hat{R}_{1:n+1}$ is strongly convex is weaker than strong convexity of the loss function itself, and allows for, e.g., ordinary least squares with a full-rank design matrix, even though the individual losses may not be strongly convex.

Combining Proposition~\ref{prop:erm-grad-stability} with Theorem~\ref{thm:main-high-d} gives a multi-dimensional risk control bound, as below.
\begin{corollary}
    \label{cor:regularized-erm}
    In the same setting as Proposition~\ref{prop:erm-grad-stability}, we have that
    \begin{equation}
        \E\left[\nabla\ell(X_{n+1}, Y_{n+1}; \thetahat_n) \right] \preceq \beta - \lambda\E[\thetahat_{n}].
    \end{equation}
\end{corollary}
This tells us the risk is controlled up to $\beta$ plus a term to handle regularization.

Finally, we present an adjusted algorithm that achieves a conservative risk control guarantee on the gradient.
Let
\begin{equation}
    \label{eq:erm-linear}
    \widetilde{\A}_\gamma(D) = \argmin_{\theta \in \R^d} \hat{R}_D(\theta) + \tfrac{\lambda}{2} \|\theta\|_2^2 + \gamma \mathbf{1}_d^{\top}\theta.
\end{equation}
The final term, $\gamma \mathbf{1}_d^{\top}\theta$, has gradient $\gamma\mathbf{1}_d$, pushing all components of the solution in the negative direction to compensate, while retaining convexity.
A careful selection of $\gamma$ to balance terms can therefore yield risk control.
\begin{proposition}
    \label{prop:regularized-erm-conservative}
    In the same setting as Proposition~\ref{prop:erm-grad-stability}, we have that 
    \begin{equation}
        \E\left[\nabla\ell(X_{n+1}, Y_{n+1}; \thetahat_n) \right] \preceq \beta_\gamma - \lambda\E[\thetahat_{n}] - \gamma\mathbf{1}_d,
    \end{equation}
    where
    \begin{equation}
        \beta_{\gamma} = \frac{\E[\rho(Z_{n+1})\|\nabla\ell(Z_{n+1};\thetahat_n)\|_2] + \E\left[\rho(Z_{n+1})\right]\left(\E\left[\|\tfrac{1}{n}\sum_{j =1}^n\nabla\ell(Z_j;\thetahat_n)\|_2\right]+2\gamma\right)}{(\mu + \lambda) (n+1)} \mathbf{1}_d.
    \end{equation}
    Furthermore, provided $n > \tfrac{2}{\mu+\lambda} \E[\rho(Z_{n+1})] - 1$, there exists an explicit choice of $\gamma$ that gives
    \begin{equation}
        \E\left[ \nabla\ell(X_{n+1}, Y_{n+1}; \thetahat_n) \right] \preceq 0.
    \end{equation}
\end{proposition}

These gradient bounds allow us to prove distribution-free multigroup risk control guarantees for ERM.
These guarantees are in the style of~\cite{gibbs2023conformal} and~\cite{blot2024automatically}, which handle the special cases of the quantile loss~\cite{koenker1978regression} and monotonic losses in $d=1$ respectively, and both of which only give full-conformal type guarantees that require having access to the test covariate $X_{n+1}$ and potentially looping through all possible test labels $y \in \cY$.
Herein, we allow a more general class of non-monotonic losses, a high-dimensional $\theta$, and also give split-conformal versions of these techniques that do not depend on the test covariate or label.
Depending on the form of the gradient, these guarantees can recover  multivalidity guarantees~\cite{bastani2022practical, jung2022batch} or multiaccuracy guarantees~\cite{kim2019multiaccuracy}, and are also related to the literature on multicalibration~\cite{hebert2018multicalibration, deng2023happymap, noarov2023statistical}.

The first example shows how we can post-process a black-box model to achieve multigroup unbiasedness using ordinary least squares.
For this purpose, define $\mathbf{X}_{n+1}=(X_1, \ldots, X_{n+1})^\top$ and $\lambda_{\rm min}$ as the minimum eigenvalue function.
\begin{corollary}
    \label{cor:unbiased-least-squares}
    Let $\cX=\{0,1\}^d$, $f : \cX \to \R$ be any fixed function, $\ell(x,y;\theta) = \tfrac{1}{2}\|f(x) + x^\top\theta - y\|_2^2$, and $\cA(D) = \argmin_{\theta \in \R^d} \frac{1}{|D|}\sum_{(x,y)\in D}\ell(x,y;\theta)$.
    Assume also that $\lambda_{\rm min}(\mathbf{X_{n+1}}^\top\mathbf{X_{n+1}}/(n+1)) \geq \mu > 0$ almost surely.
    Then for all $j \in [d]$,
    \begin{equation}
        \E[Y_{n+1} \mid X_{n+1,j} = 1] \in \left( \E[f(X_{n+1}) + X_{n+1}^\top\thetahat_{n} \mid X_{n+1,j} = 1] \pm \frac{d^{\frac{3}{2}}\E[|f(X_{n+1})+X_{n+1}^\top\thetahat_{n}-Y_{n+1}|]}{\mu(n+1)\P(X_{n+1,j}=1)}\right).
    \end{equation}
\end{corollary}
If $X_{n+1}$ represents a vector of group membership indicators, the adjusted prediction $f(X_{n+1})+X_{n+1}^\top\thetahat_n$ is approximately unbiased for all groups simultaneously, so long as they have sufficient probability of being observed.
This is essentially the batch version of the gradient equilibrium guarantees in~\cite{angelopoulos2025gradient}.
Much like in the case of gradient equilibrium, Corollary~\ref{cor:regularized-erm} extends to all generalized linear models, not just linear regression, and we omit the proof because it is essentially identical.

\subsection{Estimating the Stability Parameter}
\label{sec:beta-estimation}
Here, we seek to estimate upper bounds on the stability parameter.
We will do so using the calibration sample, and averaging across bootstrap replicates. 
Let $\hat{\P}_n$ be the empirical measure on $D_{1:n}=\{Z_i=(X_i,Y_i)\}_{i=1}^n$. For $b=1,\ldots,B$, draw a bootstrap dataset
\[
D^{(b)}=\bigl(Z^{(b)}_1,\ldots,Z^{(b)}_{n+1}\bigr)\sim\hat{\P}_n^{\,n+1},
\]
compute the estimate of $\beta$ on $D^{(b)}$, and average the bootstrapped results.
In well-behaved circumstances~\cite{vandervaart1996weak}, the bootstrap mean consistently estimates the corresponding population expectation, and we find this to be the case in our experiments.
That said, the validity of the bootstrap under subsampling is not trivial~\cite{politis1999subsampling} and represents an interesting avenue for further investigation.

The tightest way to estimate $\beta$ is directly from its definition, as opposed to applying the bounds from earlier in Section~\ref{sec:methods}.
Fix a reference algorithm $\A^*$ (as in Theorem~\ref{thm:main-intro}). For each bootstrap dataset $D^{(b)}$ compute
\begin{equation}
\label{eq:beta-def-bootstrap}
\Delta^{(b)}
\;:=\;
\frac{1}{n+1}\sum_{i=1}^{n+1}
\Bigl[\,
\ell\bigl(Z^{(b)}_i;\A(D^{(b)}_{-i})\bigr)
-
\ell\bigl(Z^{(b)}_i;\A^*(D^{(b)})\bigr)
\,\Bigr].
\end{equation}
We estimate $\beta$ by the positive part of the bootstrap mean
\[
\widehat{\beta}_{\rm def} \;:=\; \bigl(\,\overline{\Delta}\,\bigr)_+,
\qquad
\overline{\Delta}=\frac{1}{B}\sum_{b=1}^B \Delta^{(b)}.
\]

\section{Experiments}
\label{sec:experiments}
This section performs experiments to validate our method in several practical setups.
First, we run a real-data selective classification experiment on the Imagenet dataset.
Next, we run FDR and IOU control experiments on tumor segmentation.
Finally, we run a multigroup debiasing experiment of recidivism predictions.
Code to reproduce the experiments is available at \url{https://github.com/aangelopoulos/nonmonotonic-crc}.

We run experiments with three methods. 
The first, which we call CRC-C, is the conservative version of conformal risk control, run at an adjusted level $\alpha'=\alpha-\widehat{\beta}_{\rm def}$, where $\widehat{\beta}_{\rm def}$ is the estimated algorithmic stability from Section~\ref{sec:beta-estimation} via the bootstrap.
The second, which we simply call CRC, is the unadjusted version of conformal risk control, i.e.,~\eqref{eq:leftmost-root} with $\alpha$ set to be the desired risk level.
Finally, we run the Learn-then-Test procedure from~\cite{angelopoulos2025learn}, which provides a distribution-free guarantee that
\begin{equation}
    \P\left(\E\left[ \ell(X_{n+1}, Y_{n+1}; \thetahat) \big\vert D_{1:n} \right] \leq \alpha\right) \geq 1-\delta,
\end{equation}
where here we choose $\delta=0.1$.
In all the experiments below, LTT tends to be more conservative than CRC and CRC-C, which is expected given that it is a high-probability guarantee.
Nonetheless, we include it as the standard baseline for non-monotonic risk control, and as a contrast to our expectation control method.

\subsection{Selective Classification: Imagenet}
\begin{figure}[ht]
    \centering
    \includegraphics[width=\linewidth]{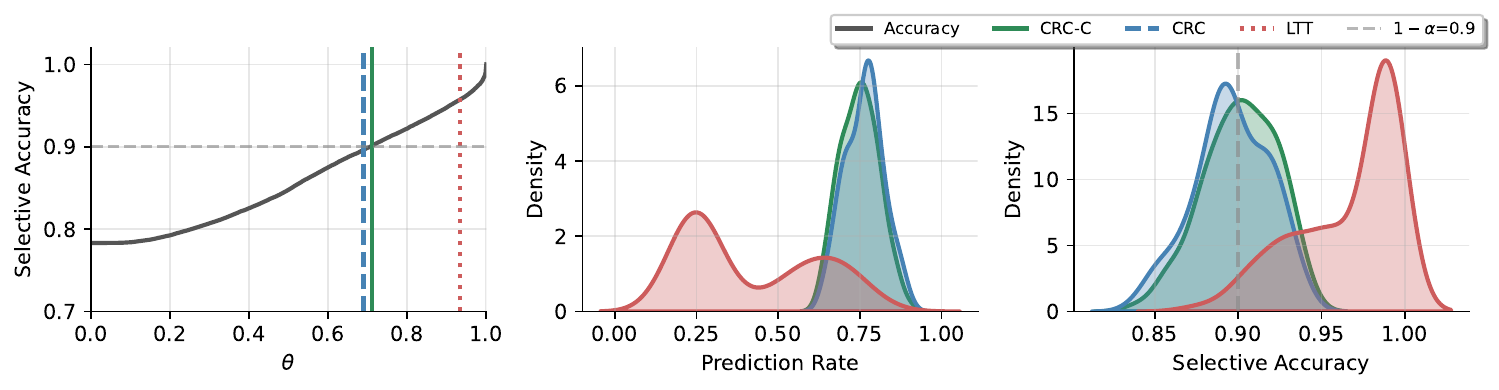}
    \caption{\textbf{Selective classification results on Imagenet.} The left panel shows the choices of thresholds from CRC, CRC-C, and LTT superimposed with the population accuracy curve. The middle panel shows a kernel density estimate of the prediction rate for all three methods over 100 resamplings of the data. The right panel shows a similar plot for the selective accuracy.}
    \label{fig:selective-imagenet}
\end{figure}
The Imagenet~\cite{deng2009imagenet} experiment follows the same setting as in~\cite{angelopoulos2023gentle}: we use a ResNet-152~\cite{he2016deep} as the base classifier, and the highest predicted probability from the model is used as $\hat{P}_i$ for each sample $i$.
The corresponding $\hat{Y}_1, \ldots, \hat{Y}_{n+1}$ represent class label predictions, and $Y_1, \ldots, Y_{n+1}$ represent the true labels.
We set $n=1000$.

To assess the algorithms' performance, we display the thresholds chosen by all three algorithms and the resulting selective accuracy and prediction rate in Figure~\ref{fig:selective-imagenet}.
The ideal algorithm issues as many predictions as possible subject to the accuracy constraint; that is, it should generally meet the accuracy constraint exactly, and maximize the prediction rate. CRC is the least conservative, LTT is the most (as it is a high-probability guarantee), and CRC-C is in the middle.
The estimated stability parameter is $\beta=0.006$, meaning CRC is essentially safe to use with no correction in this specific example.
Without correction, CRC is slightly liberal, with an accuracy a hair below $90\%$. CRC-C lands slightly above, while LTT is highly conservative and also higher variance.

\subsection{Smooth Losses: FDR in Tumor Segmentation}
\label{sec:fdr-tumor}
\begin{figure}[ht]
    \centering
    \includegraphics[width=\linewidth]{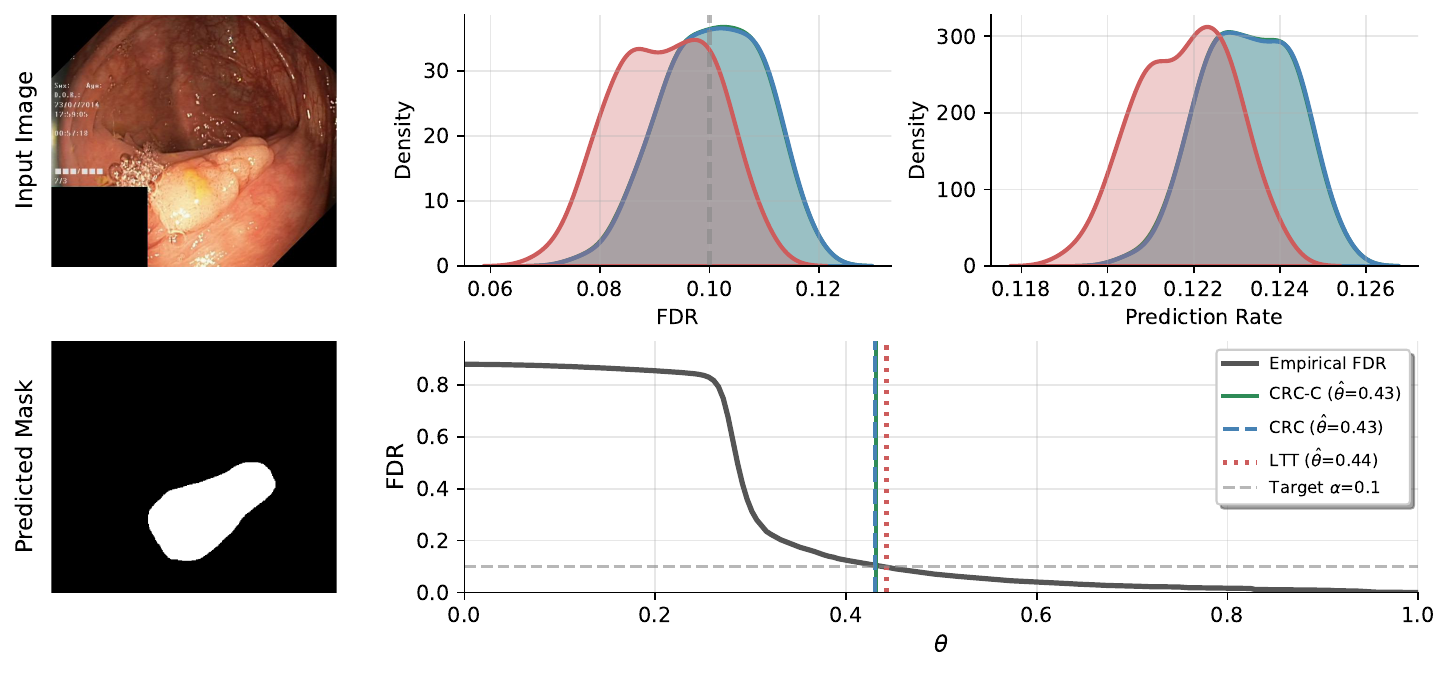}
    \caption{\textbf{FDR control results in polyp segmentation.} The top left and bottom left panels show examples of polyp images and predicted masks from the PraNet, respectively.
    The top middle and right panels show histograms of the FDR and prediction rate for CRC-C, CRC, and LTT, over 100 resamplings of the data, respectively. Notice that CRC and CRC-C overlap almost entirely.
    The bottom right panel shows selections of $\hat{\theta}$ made by the three methods on the calibration data superimposed with the true FDR calculated over the whole dataset.}
    \label{fig:fdr-tumor}
\end{figure}
\begin{figure}[ht]
    \centering
    \includegraphics[width=\linewidth]{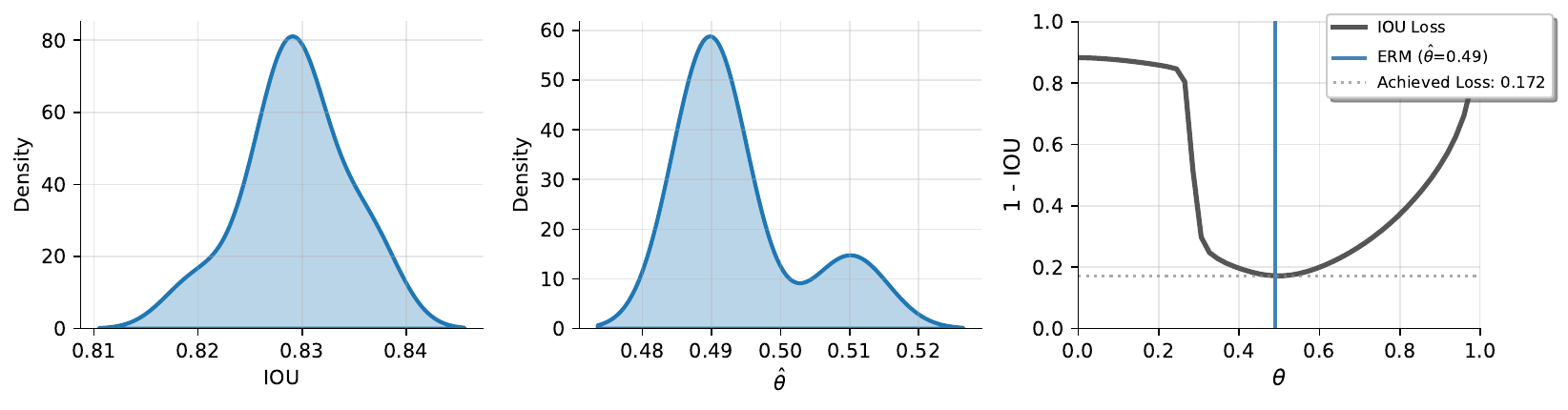}
    \caption{\textbf{IOU control results in polyp segmentation.} The left and middle panels show histograms of the IOU and segmentation parameter $\theta$ for ERM over 100 resamplings of the data, respectively.
    The right panel shows the selection of $\hat{\theta}$ superimposed with the true IOU calculated over the whole dataset.}
    \label{fig:iou-tumor}
\end{figure}
Next, we handle false discovery rate~\cite{benjamini1995controlling} (FDR) control on the PraNet~\cite{fan2020pranet} polyp segmentation task from~\cite{angelopoulos2023gentle}.
Each image $X_i$ is a $d \times d$ grid of pixels, and each label $Y_i$ is a $d \times d$ binary grid, with `1` signifying the presence of a tumor. The predictive model, $f(X_i) \in [0,1]$, outputs a probability that each pixel contains a tumor, which we binarize at level $\theta$ to form $\hat{y}(x;\theta)$. 
There are 1798 examples total in the dataset.
The FDR loss is then
\begin{equation}
    \ell(x,y;\theta) = 1 - \frac{|y \odot \hat{y}(x;\theta)|}{|\hat{y}(x;\theta)|},
\end{equation}
where the notation $|y|$ represents the sum over all values in $y$ and $\odot$ represents the element-wise product.
The FDR is non-monotonic in $\theta$, since both the numerator and denominator grow as $\theta$ shrinks. However, it tends to be a smooth and well-behaved loss function.
We generally seek to maximize discoveries subject to our FDR constraint, so we pick the largest set (smallest $\theta$) that controls our risk: $\A(D) = \inf\{\theta : \tfrac{1}{|D|}\sum_{(x,y) \in D} \ell(x,y;\theta) \leq \alpha\}$.

Figure~\ref{fig:fdr-tumor} shows empirical results from the FDR segmentation task with $n=500$.
Over 1000 bootstrap replicates, we found $\beta\approx0.00007$ for this task, indicating that correction is not needed for essentially any practical purpose due to the regularity of this loss function.
The results bear this out, with the results of CRC-C and CRC overlapping almost entirely, and LTT being more conservative.

\subsection{ERM Loss Guarantees: IOU in Tumor Segmentation}
Next, we handle intersection-over-union (IOU) control in the same tumor segmentation setting as Section~\ref{sec:fdr-tumor}.
The IOU is a measure of overlap common in the image segmentation literature, defined as the cardinality of the intersection of two segmentation masks divided by the cardinality of their union. 
In the perfect case, the IOU is 1, but it can degrade if the masks do not perfectly overlap.
Ideally, we want to maximize the IOU, making this problem a good candidate for empirical risk minimization.

We formally define the IOU loss function as
\begin{equation}
    \ell(x,y;\theta) = 1 - \frac{|y \odot \hat{y}(x;\theta)|}{|\max(y, \hat{y})|},
\end{equation}
where the denominator denotes the cardinality of the element-wise maximum.
Our algorithm $\cA$ minimizes the empirical risk with respect to this loss.

Results in Figure~\ref{fig:iou-tumor} indicate that the ERM procedure achieves an approximate minimizer. We calculated $\beta=0.000056$ for this problem over 100 bootstrap iterations.

\subsection{ERM Gradient Guarantees: Multigroup Debiasing of Recidivism Predictions}
\begin{figure}[ht]
    \centering
    \includegraphics[width=0.7\linewidth]{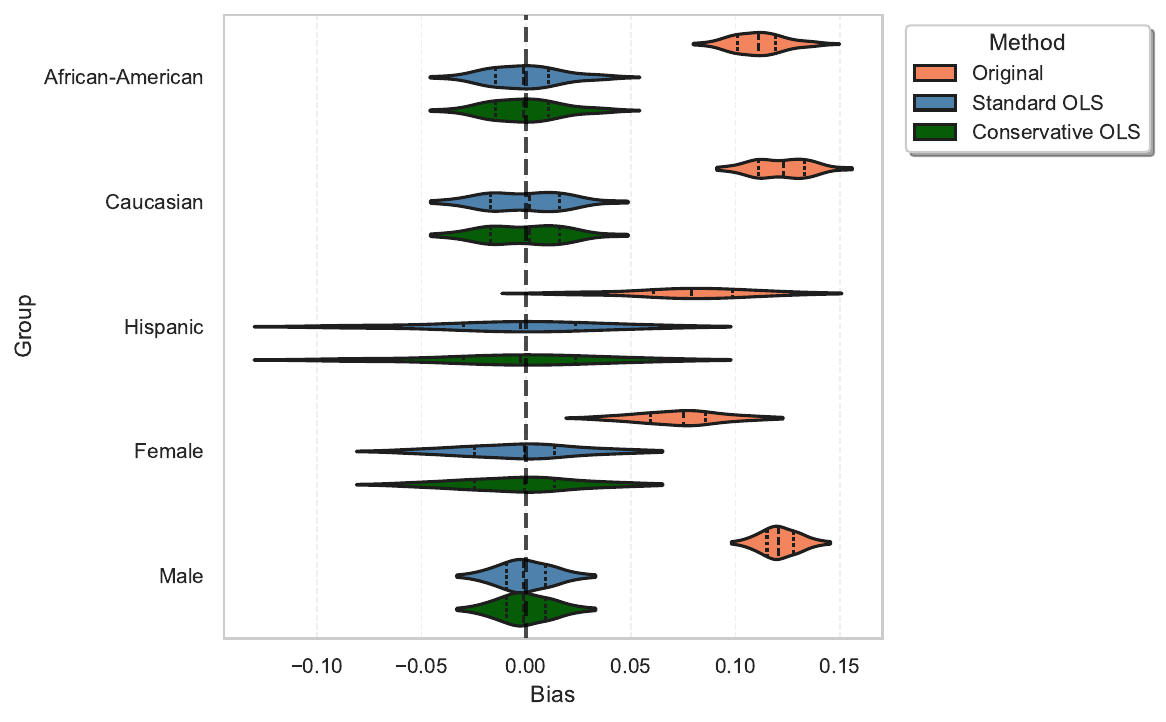}
    \caption{\textbf{Results of multigroup debiasing on COMPAS dataset.} Each row is a group, and the horizontal axis shows the amount of predictive bias. Each violin plot shows the bias over 100 resamplings of the data for each method.}
    \label{fig:compas}
\end{figure}
Finally, we use the ERM algorithm to issue unbiased recidivism predictions over racial groups and sexes on the Correctional Offender Management Profiling for Alternative Sanctions (COMPAS) Recidivism and Racial Bias dataset~\cite{angwin2022machine}, similarly to~\cite{angelopoulos2025gradient}. 
Focusing on arrest records from Broward County, Florida, this dataset tracks demographic details and criminal backgrounds alongside COMPAS-generated recidivism risk scores.
These scores are utilized to inform pretrial bail decisions, but the errors exhibit significant racial disparities. 
The labels $Y_i \in \{0,1\}$ are binary indicators of recidivism.
The covariates $X_i \in \{0,1\}^5$ are a vector of indicators for each group: Black, White, Hispanic, Male, and Female, respectively.
Note that these groups overlap.
The predicted recidivism rate $f_i$ is obtained by taking the COMPAS recidivism score, which lies on an integer scale from 1 to 10, and dividing it by 10. 
We preprocess the dataset to remove rare groups, and only consider the Caucasian, African-American, and Hispanic subgroups along with Male and Female sexes.
We take $n=1000$ from the full dataset of 6787 samples.

We aim to recalibrate $f_i$ with a regression as in Corollary~\ref{cor:unbiased-least-squares}.
We do this by running a post-hoc conservative OLS on the residual of the COMPAS predictor (with an optional linear term to provide a conservative guarantee): i.e., taking $\A$ to be ERM with
\begin{equation}
    \ell(z;\theta) = \frac{1}{2}\big(x^\top\theta - (y-f)\big)^2 + \beta \mathbf{1}_5^\top\theta.
\end{equation}
The implied guarantee says that the adjusted predictions, $f_i+X_i^\top\thetahat_n$, are unbiased for all races and sexes, as stated in Corollary~\ref{cor:unbiased-least-squares}. 
We conduct experiments evaluating the bias in Figure~\ref{fig:compas} using three methods: the raw COMPAS predictor, the standard OLS with no linear term, and the conservative OLS including the linear term.
In this case, $\beta=0.001839$, so there is almost no adjustment needed in this specific example.
The figure shows that the method effectively ensures all groups are unbiased.

\section{Discussion}
This paper showed that the validity of conformal risk control can be viewed as a consequence of algorithmic stability, and that therefore, all stable algorithms benefit from conformal guarantees.
This expands the scope of the field, allowing us to prove distribution-free bounds for arbitrary tasks.
Finding such tasks and corresponding guarantees would be an exciting avenue for future work.

The insights in this work also extend to full-conformal-type algorithms (see ~\cite{vovk2005algorithmic} and the extension to risk control in~\cite{angelopoulos2024note}). 
The main difference is that the definition of stability becomes leave-one-label-out, as opposed to leave-one-datapoint-out, and the algorithm can therefore involve imputing the missing label.
Otherwise, the main idea is the same.
Similarly, the technique extends readily to non-exchangeable distributions (see~\cite{tibshirani2019conformal,barber2022conformal} and Chapter 7 of~\cite{angelopoulos2024theoretical} for a review) via weighted exchangeability.
We leave further exploration of these important themes to future work.

It is also worth noting that the bounds developed in Section~\ref{sec:methods} are not necessarily comprehensive or tight.
It is possible that tighter or more general stability bounds exist for these algorithms.
As an example, even irregular loss functions can be made continuous and Lipschitz via randomized smoothing, leading to a natural extension of Proposition~\ref{prop:stability-smooth}.
Future work could develop a more unified theory for developing stability bounds for conformal-type algorithms, similarly to that on ERM.
Proving guarantees for the bootstrap estimate of $\beta$ would also be an exciting avenue for future work, as it is more practical than applying concentration arguments to the bounds in Section~\ref{sec:experiments}, hence why we use it in our experiments.

Finally, there is ample work to be done connecting this work with other ideas in statistics and machine learning.
There is a clear connection to the algorithmic stability literature that should be fully explored; additional areas of interest include multicalibration and its connections to our gradient-control guarantees, stochastic gradient descent as a potential algorithm, black-box tests of stability in the style of~\cite{kim2023black} to outline when conformal-type guarantees are (im)possible, and constrained optimization which may also lead to interesting risk-control guarantees; as an example, the parameter may be constrained to lie on the simplex.

\subsection*{Acknowledgments}
I would like to thank my friends and colleagues for their valuable comments and feedback on this paper: Rina Foygel Barber, Stephen Bates, Tiffany Ding, John Duchi, Yaniv Romano, Ryan Tibshirani, Vladimir Vovk, my brilliant wife Tijana Zrni\'c Angelopoulos, and the anonymous reviewers from the \emph{Philosophical Transactions of the Royal Society}.

\clearpage
\appendix

\section{Additional Proofs and Results}
\begin{proof}[Proof of Proposition~\ref{prop:general-stability}]
Since we have no assumptions on the loss, we will rely on a relatively coarse bound,
\begin{equation}
    \label{eq:basic-bound-generic}
    \E[R(\thetahat_{-i})] \leq \alpha + \epsilon + \P\left(R(\thetahat_{-i}) > \alpha + \epsilon \right) \text{ for all } \epsilon \geq 0.
\end{equation}
We will take the strategy of union-bounding the probability on the right-hand side and deriving the optimal $\epsilon$.
Fix $\epsilon\ge 0$, $i \in [n+1]$, and define the random set
\begin{equation}
\mathcal{B}_\epsilon\ :=\ 
\Bigl\{\theta\in\Theta_m:
\ R(\theta)>\alpha+\epsilon\ \text{ and }\ 
\hat R_{-i}(\theta)\le \alpha
\Bigr\}.
\end{equation}
Since $\hat R_{-i}(\thetahat_{-i})\le \alpha$,
\begin{equation}
\{R(\thetahat_{-i})>\alpha+\epsilon\}\ \subseteq\ \{\mathcal{B}_\epsilon\neq\varnothing\},
\end{equation}
hence
\begin{equation}
\P\bigl(R(\thetahat_{-i})>\alpha+\epsilon\bigr)\ \leq\ \P(\mathcal{B}_\epsilon\neq\varnothing).
\end{equation}
Now,
\begin{equation}
\{\mathcal{B}_\epsilon\neq\varnothing\}
=\bigcup_{\theta\in\Theta_m}
\Bigl\{R(\theta)>\alpha+\epsilon,\ \hat R_{-i}(\theta)\le \alpha\Bigr\},
\end{equation}
so by a union bound,
\begin{equation}
\P(\mathcal{B}_\epsilon\neq\varnothing)
\ \leq\ 
\sum_{\theta\in\Theta_m}
\P\Bigl(R(\theta)>\alpha+\epsilon,\ \hat R_{-i}(\theta)\le \alpha\Bigr).
\end{equation}
Fix $\theta\in\Theta_m$. If $R(\theta)\le \alpha+\epsilon$ the term is zero.  
If $R(\theta)>\alpha+\epsilon$, then
\begin{equation}
\P\Bigl(R(\theta)>\alpha+\epsilon,\ \hat R_{-i}(\theta)\le \alpha\Bigr)
\ \leq\ 
\P\bigl(\hat R_{-i}(\theta)-R(\theta)\le -\epsilon\bigr)
\ \leq\ 
\exp(-2n\epsilon^2),
\end{equation}
by Hoeffding’s inequality. Therefore,
\begin{equation}
\P(\mathcal{B}_\epsilon\neq\varnothing)
\ \leq\ 
(m+1)\exp(-2n\epsilon^2),
\end{equation}
which combined with~\eqref{eq:basic-bound-generic} means, for all $\epsilon \geq 0$,
\begin{equation}
\label{eq:general-bound-epsilon}
\E[R(\thetahat_{n})] \leq \alpha+\epsilon+(m+1)\exp(-2n\epsilon^2).
\end{equation}

Next, we optimize over $\epsilon \geq 0$.
The optimal $\epsilon^*$ of the right-hand side of~\eqref{eq:general-bound-epsilon} solves the first-order condition
\begin{align}
    \label{eq:general-bound-simplified}
    1+(m+1)(-4n\epsilon^*)\exp(-2n(\epsilon^*)^2)=0 & \Longleftrightarrow (m+1)\exp(-2n(\epsilon^*)^2) = \frac{1}{4n\epsilon^*} \\
    & \Longleftrightarrow (\epsilon^*)^2\exp(-4n(\epsilon^*)^2) = \frac{1}{16n^2(m+1)^2} \\
    & \Longleftrightarrow -4n(\epsilon^*)^2\exp(-4n(\epsilon^*)^2) = -\frac{1}{4n(m+1)^2} \\
    & \Longleftrightarrow -2ue^{-2u} = -\frac{1}{4n(m+1)^2} \\
    & \Longleftrightarrow -\frac{1}{2} W_{-1} \left( -\frac{1}{4n(m+1)^2}\right) = u \\
    & \Longleftrightarrow \epsilon^* = 
    \sqrt{-\frac{W_{-1} \left( -\frac{1}{4n(m+1)^2}\right)}{4n}}
\end{align}
where $u=2n(\epsilon^*)^2$.
Notice that~\eqref{eq:general-bound-simplified} proves the result.
To show it is $\tilde{\mathcal{O}}(n^{-1/2})$, it is known~\cite{loczi2020explicit} that $W_{-1}(x) \geq \frac{e\ln(-x)}{e-1}$ for $x \in [-1/e, 0)$, allowing us to write
\begin{equation}
    \epsilon^* \leq 
    \sqrt{-\frac{W_{-1} \left( -\frac{1}{4n(m+1)^2}\right)}{4n}} \leq \sqrt{\frac{e\ln(4n(m+1)^2)}{4(e-1)n}}.
\end{equation}
In the same range, $W_{-1}(x) \leq \ln(-x)-\ln(-\ln(-x))$, yielding
\begin{equation}
    \epsilon^* \geq \sqrt{\frac{\ln(4n(m+1)^2)+\ln(\ln(4n(m+1)^2))}{4n}}.
\end{equation}
Plugging in these bounds into~\eqref{eq:general-bound-simplified} gives the result in terms of big-O notation.
\end{proof}

\begin{proof}[Proof of Proposition~\ref{prop:erm-stable-loss}]
By~\eqref{eq:lipschitz-rho-simple}, we have that
\begin{equation}
    \frac{1}{n+1}\sum_{i=1}^{n+1}\left(\ell(Z_i;\thetahat_{-i}) - 
    \ell(Z_i;\thetahat_{n+1})\right) \leq \frac{1}{n+1}\sum_{i=1}^{n+1} \rho(Z_i)\| \thetahat_{n+1} - 
    \thetahat_{-i} \|_2.
\end{equation}
The remainder of this proof will show that the following upper-bound holds:
\begin{equation}
    \label{eq:erm-rho-ub}
    \frac{1}{n+1}\sum_{i=1}^{n+1} \rho(Z_i)\| \thetahat_{n+1} - 
    \thetahat_{-i} \|_2 \leq \frac{2}{\lambda (n+1)^2}\sum_{i=1}^{n+1}\rho(Z_i)^2,
\end{equation}
after which taking expectations reveals that
\begin{equation}
    \E\left[\frac{1}{n+1}\sum_{i=1}^{n+1}\left(\ell(Z_i;\thetahat_{-i}) - 
    \ell(Z_i;\thetahat_{n+1})\right)\right] \leq \frac{2}{\lambda (n+1)^2}\sum_{i=1}^{n+1}\E\left[\rho(Z_i)^2\right] = \frac{2\E\left[\rho(Z_1)^2\right]}{\lambda (n+1)},
\end{equation}
proving the result.

Towards proving~\eqref{eq:erm-rho-ub}, we will upper-bound $\|\thetahat_{n+1} - \thetahat_{-i}\|$ in terms of $\rho(Z_i)$. We will make use of the fact that for any $g\in\partial_\theta \ell(z;\theta)$,
\begin{equation}
    \label{eq:subgradient-rho}
    \ell(z;\theta)+\|g\|_2^2 \leq \ell(z;\theta+g) \leq \ell(z;\theta)+\rho(z)\|g\|_2 \implies \|g\|_2 \leq \rho(z).
\end{equation}

Define the (strongly convex) objectives
\begin{equation}
F(\theta)\ :=\ \hat R_{D_{1:n+1}}(\theta)+\frac{\lambda}{2}\|\theta\|_2^2,
\qquad
F_{-i}(\theta)\ :=\ \hat R_{D_{-i}}(\theta)+\frac{\lambda}{2}\|\theta\|_2^2.
\end{equation}
Since $0\in\partial F(\thetahat_{n+1})$, there exist $g_1,\ldots,g_{n+1}$ with $g_j\in\partial_\theta \ell(Z_j;\thetahat_{n+1})$ such that
\begin{equation}
\label{eq:opt-full}
\frac{1}{n+1}\sum_{j=1}^{n+1} g_j\ +\ \lambda \thetahat_{n+1}\ =\ 0.
\end{equation}
Define
\begin{equation}
s_{-i}\ :=\ \frac{1}{n}\sum_{j\neq i} g_j\ +\ \lambda \thetahat_{n+1}.
\end{equation}
Since $\hat R_{D_{-i}}$ is an average of convex functions, $\frac{1}{n}\sum_{j\neq i}g_j\in\partial \hat R_{D_{-i}}(\thetahat_{n+1})$, hence $s_{-i}\in\partial F_{-i}(\thetahat_{n+1})$.
Using~\eqref{eq:opt-full},
\begin{equation}
s_{-i}
=
\frac{1}{n}\sum_{j\neq i} g_j\ -\ \frac{1}{n+1}\sum_{j=1}^{n+1} g_j
=
\frac{1}{(n+1)n}\sum_{j\neq i} g_j\ -\ \frac{1}{n+1}g_i.
\end{equation}

We now use the strong convexity of $F_{-i}$.
Take any $u,v\in\R^d$, and any $a\in\partial F_{-i}(u)$ and $b\in\partial F_{-i}(v)$.
Then $a=a_0+\lambda u$ and $b=b_0+\lambda v$ for some $a_0\in\partial\hat{R}_{-i}(u)$ and $b_0\in\partial\hat{R}_{-i}(v)$.
By monotonicity of the subdifferential operator, $\langle a_0-b_0,\ u-v\rangle\ \ge\ 0$, so
\begin{equation}
\label{eq:strong-mono}
\langle a-b,\ u-v\rangle
=
\langle a_0-b_0,\ u-v\rangle + \lambda\|u-v\|_2^2
\ \ge\ \lambda\|u-v\|_2^2.
\end{equation}
Applying~\eqref{eq:strong-mono} with $u=\thetahat_{n+1}$, $v=\thetahat_{-i}$, $a=s_{-i}\in\partial F_{-i}(\thetahat_{n+1})$, and $b=0\in\partial F_{-i}(\thetahat_{-i})$ gives
\begin{equation}
\lambda\|\thetahat_{n+1}-\thetahat_{-i}\|_2^2\ \le\ \langle s_{-i},\ \thetahat_{n+1}-\thetahat_{-i}\rangle
\ \le\ \|s_{-i}\|_2\,\|\thetahat_{n+1}-\thetahat_{-i}\|_2,
\end{equation}
hence, by~\eqref{eq:subgradient-rho}, 
\begin{align}
\label{eq:param-shift}
\rho(Z_i)\|\thetahat_{n+1}-\thetahat_{-i}\|_2\ & \leq \frac{\rho(Z_i)}{\lambda}\|s_{-i}\|_2 \\
& = \frac{\rho(Z_i)}{\lambda}\left\|\frac{1}{(n+1)n}\sum_{j\neq i} g_j\ -\ \frac{1}{n+1}g_i\right\|_2 \\
& \leq \frac{1}{\lambda (n+1)}\left(\rho(Z_i)^2+\frac{1}{n}\sum_{j\neq i}\rho(Z_i)\rho(Z_j)\right).
\end{align}
Averaging over $i=1,\ldots,n+1$ yields
\begin{align}
\label{eq:avg-diff}
\frac{1}{n+1}\sum_{i=1}^{n+1}\rho(Z_i)\|\thetahat_{n+1}-\thetahat_{-i}\|_2\
&\le
\frac{1}{\lambda (n+1)}\left(
\frac{1}{n+1}\sum_{i=1}^{n+1}\rho(Z_i)^2
+\frac{1}{(n+1)n}\sum_{i\neq j}\rho(Z_i)\rho(Z_j)
\right).
\end{align}
Let $S:=\sum_{i=1}^{n+1}\rho(Z_i)$ and $Q:=\sum_{i=1}^{n+1}\rho(Z_i)^2$.
Then $\sum_{i\neq j}\rho(Z_i)\rho(Z_j)=S^2-Q$, so the second term in~\eqref{eq:avg-diff} equals $(S^2-Q)/((n+1)n)$.
By Cauchy--Schwarz, $S^2\le (n+1)Q$, hence
\begin{equation}
\frac{S^2-Q}{(n+1)n}\ \le\ \frac{(n+1)Q-Q}{(n+1)n}\ =\ \frac{Q}{n+1}.
\end{equation}
Plugging this into~\eqref{eq:avg-diff} gives the desired bound in~\eqref{eq:erm-rho-ub}.
\end{proof}

\begin{proof}[Proof of Proposition~\ref{prop:erm-grad-stability}]
Define
\begin{equation}
F_{n+1}(\theta) = \frac{1}{n+1}\sum_{j=1}^{n+1} \ell(z_j;\theta) + \frac{\lambda}{2}\|\theta\|_2^2,
\qquad
\hat{\theta}_{n+1}\in\arg\min_{\theta}F_{n+1}(\theta)=\A(D_{1:n+1}),
\end{equation}
and $F_{-i},\hat{\theta}_{-i}$ accordingly for all $i$. The $\mu$-strong convexity of the risks implies that $F_{n+1}$ is $(\mu+\lambda)$-strongly convex.
This further implies that for any $u\in\R^d$, $\|\hat{\theta}_{n+1}-u\|_2 \leq \frac{1}{\mu+\lambda}\|\nabla F_{D_{1:n+1}}(u)\|_2$.
Applying this with $u=\hat{\theta}_{-i}$ gives
\begin{equation}
\|\hat{\theta}_{n+1}-\hat{\theta}_{-i}\|_2
\leq \frac{1}{\mu+\lambda}\|\nabla F_{D_{1:n+1}}(\hat{\theta}_{-i})\|_2 = \frac{1}{\mu+\lambda}\|\nabla F_{D_{1:n+1}}(\hat{\theta}_{-i}) - \nabla F_{-i}(\hat{\theta}_{-i})\|_2,
\end{equation}
where in the last step we used that $0 = \nabla F_{-i}(\hat{\theta}_{-i})$. 
Next, we use the fact that 
\begin{multline}
\nabla F_{D_{1:n+1}}(\hat{\theta}_{-i}) - \nabla F_{-i}(\hat{\theta}_{-i}) = \frac{1}{n+1}\nabla\ell(Z_i;\hat{\theta}_{-i}) + \Big(\frac{1}{n+1}-\frac{1}{n}\Big)\sum_{j\neq i} \nabla\ell(Z_j;\hat{\theta}_{-i}) \\
= \frac{1}{n+1}\nabla\ell(Z_i;\hat{\theta}_{-i}) - \frac{1}{n(n+1)}\sum_{j\neq i} \nabla\ell(Z_j;\hat{\theta}_{-i}),
\end{multline}
so $\|\hat{\theta}_{n+1}-\hat{\theta}_{-i}\|_2 \leq \tfrac{1}{\mu + \lambda}\left(\tfrac{1}{n+1}\|\nabla\ell(Z_i;\hat{\theta}_{-i})\|_2 + \tfrac{1}{n(n+1)}\left\|\sum_{j\neq i} \nabla\ell(Z_j;\hat{\theta}_{-i})\right\|_2\right)$.

Next, picking $i=n+1$ and using $\rho(z)$-Lipschitzness,
\begin{multline}
\|\nabla\ell(Z_{n+1};\thetahat_n)-\nabla\ell(Z_{n+1};\hat{\theta}_{n+1})\|_2
\leq
\rho(Z_{n+1})\|\hat{\theta}_{n+1}-\thetahat_n\|_2
\\
\leq \frac{\rho(Z_{n+1})\|\nabla\ell(Z_{n+1};\thetahat_n)\|_2}{(\mu+\lambda)(n+1)} + \frac{\rho(Z_{n+1})\|\sum_{j=1}^{n} \nabla\ell(Z_j;\thetahat_n)\|_2}{(\mu+\lambda) n(n+1)}.
\end{multline}
Taking expectations gives
\begin{multline}
\E\left[\|\nabla\ell(Z_{n+1};\thetahat_n)-\nabla\ell(Z_{n+1};\hat{\theta}_{n+1})\|_2\right] \\
\leq \frac{\E\left[\rho(Z_{n+1})\|\nabla\ell(Z_{n+1};\thetahat_n)\|_2\right]}{(\mu+\lambda)(n+1)} + \frac{\E\left[\rho(Z_{n+1})\|\tfrac{1}{n}\sum_{j=1}^{n} \nabla\ell(Z_j;\thetahat_n)\|_2\right]}{(\mu+\lambda) (n+1)}.
\end{multline}
Finally, using norm inequalities gives
\begin{align}
    & \E\left[\nabla\ell(Z_{n+1};\thetahat_n) - \nabla\ell(Z_{n+1};\hat{\theta}_{n+1})\right] \\
    \preceq &\| \E\left[\nabla\ell(Z_{n+1};\thetahat_n) - \nabla\ell(Z_{n+1};\hat{\theta}_{n+1})\right] \|_{\infty} \mathbf{1}_d \\
    \preceq &\E\left[\|\nabla\ell(Z_{n+1};\thetahat_n) - \nabla\ell(Z_{n+1};\hat{\theta}_{n+1})\|_2\right] \mathbf{1}_d  \\
    \preceq & \frac{\E\left[\rho(Z_{n+1})\|\nabla\ell(Z_{n+1};\thetahat_n)\|_2\right] + \E\left[\rho(Z_{n+1})\|\tfrac{1}{n}\sum_{j=1}^{n} \nabla\ell(Z_j;\thetahat_n)\|_2\right]}{(\mu+\lambda)(n+1)}\mathbf{1}_d,
\end{align}
implying the desired conclusion.
\end{proof}
\begin{proof}[Proof of Corollary~\ref{cor:regularized-erm}]
    First notice that
    \begin{equation}
        \frac{1}{n+1}\sum\limits_{i=1}^{n+1}\nabla\ell(X_{i}, Y_{i}; \thetahat_{n+1}) + \lambda\thetahat_{n+1} = 0,
    \end{equation}
    so by exchangeability,
    \begin{equation}
        \E\left[\nabla\ell(Z_{n+1},\thetahat_{n+1}) + \lambda\thetahat_{n+1}\right] = \E\left[ \frac{1}{n+1}\sum\limits_{i=1}^{n+1}\nabla\ell(X_{i}, Y_{i}; \thetahat_{n+1}) + \lambda\thetahat_{n+1} \right] = 0.
    \end{equation}
    Combining Proposition~\ref{prop:erm-grad-stability} and Theorem~\ref{thm:main-high-d} therefore gives that
    \begin{equation}
        \E\left[\nabla\ell(Z_{n+1},\thetahat_{n})\right] \preceq \beta - \lambda\E[\thetahat_{n}].
    \end{equation}
\end{proof}
\begin{proof}[Proof of Proposition~\ref{prop:regularized-erm-conservative}]
    The same argument as Corollary~\ref{cor:regularized-erm} applied to the loss function $\ell(z;\theta) + \gamma\mathbf{1}_d^\top\theta$ gives that
    \begin{equation}
        \E\left[\nabla\ell(X_{n+1}, Y_{n+1}; \thetahat_{n}) \right] \preceq \beta_\gamma - \lambda\E[\thetahat_n] - \gamma\mathbf{1}_d.
    \end{equation}
    The smallest scalar $\gamma$ that makes the right-hand side zero is
    \begin{equation}
        \gamma
        =
        \frac{
        \displaystyle
        \frac{
        \E\left[\rho(Z_{n+1})\|\nabla\ell(Z_{n+1};\thetahat_{n}\|_2\right]
        +
        \E[\rho(Z_{n+1})]\,
        \E\left[\left\|\frac{1}{n}\sum_{j=1}^n \nabla\ell(Z_j;\thetahat_{n})\right\|_2\right]
        }{(\mu+\lambda)(n+1)}
        +
        \lambda\left\|\E[\thetahat_{n}]\right\|_\infty
        }{
        \displaystyle
        1 - \frac{2\,\E[\rho(Z_{n+1})]}{(\mu+\lambda)(n+1)}
        },
    \end{equation}
    provided $(\mu + \lambda)(n+1) > 2 \E[\rho(Z_{n+1})]$.
\end{proof}

\begin{proof}[Proof of Corollary~\ref{cor:unbiased-least-squares}]
    Let $\thetahat_n=\A(D_{1:n})$. Here $\lambda=0$ and $\rho(Z_{n+1}) = d$, so Corollary~\ref{cor:regularized-erm} gives
    \begin{multline}
        \E[X_{n+1}(f(X_{n+1})+X_{n+1}^\top\thetahat_n - Y_{n+1})] \preceq \frac{d\E[\|X_{n+1}(f(X_{n+1})+X_{n+1}^\top\thetahat_n - Y_{n+1})\|_2]}{\mu(n+1)}\mathbf{1}_d \\
        \preceq \frac{d^{\frac{3}{2}} \E[|f(X_{n+1})+X_{n+1}^\top\thetahat_n - Y_{n+1}|]}{\mu(n+1)}\mathbf{1}_d. 
    \end{multline}
    Since $X_{n+1} \in \{0,1\}^d$, for the $j$th coordinate,
    \begin{equation}
        \E[X_{n+1,j}(f(X_{n+1})+X_{n+1}^\top\thetahat_n - Y_{n+1})] = \P(X_{n+1,j}=1)\E[(f(X_{n+1})+X_{n+1}^\top\thetahat_n - Y_{n+1})\mid X_{n+1,j}=1].
    \end{equation}
    Therefore, the earlier inequality implies that
    \begin{equation}
        \E[Y_{n+1}\mid X_{n+1,j}=1] \geq \E[f(X_{n+1})+X_{n+1}^\top\thetahat_n \mid X_{n+1,j}=1] - \frac{d^{\frac{3}{2}}\E[|f(X_{n+1})+X_{n+1}^\top\thetahat_n - Y_{n+1}|]}{\P(X_{n+1,j}=1)\mu(n+1)}.
    \end{equation}
    The symmetric argument (with all inequalities flipped, beginning with Theorem~\ref{thm:main-high-d}) gives that
    \begin{equation}
        \E[Y_{n+1}\mid X_{n+1,j}=1] \leq \E[f(X_{n+1})+X_{n+1}^\top\thetahat_n \mid X_{n+1,j}=1] + \frac{d^{\frac{3}{2}}\E[|f(X_{n+1})+X_{n+1}^\top\thetahat_n - Y_{n+1}|]}{\P(X_{n+1,j}=1)\mu(n+1)}.
    \end{equation}
\end{proof}

\bibliographystyle{alpha}
\bibliography{bibliography}

\end{document}

%% file: macros.tex
\usepackage{amssymb,amsmath,amsthm,bbm}
\usepackage{verbatim,float,url,dsfont}
\usepackage{graphicx,subcaption,psfrag}
\usepackage{algorithm,algorithmic}
\usepackage{mathtools,enumitem}
\usepackage{multirow}
\usepackage{ragged2e}
\usepackage{xr-hyper}
\usepackage{array}

\usepackage[colorlinks=true,citecolor=blue,urlcolor=blue,linkcolor=blue]{hyperref}
\usepackage[margin=1in]{geometry}
\usepackage[authoryear]{natbib}

\usepackage[utf8]{inputenc} 
\usepackage[T1]{fontenc}    
\usepackage{booktabs}       
\usepackage{nicefrac}         
\usepackage{microtype}      

\usepackage[dvipsnames]{xcolor} 
\usepackage{pifont} 
\usepackage{etoc} 
\usepackage{tikz} 
\usepackage{pgfplots}  
\pgfplotsset{compat=1.16}  

\ifdefined\TimesFont 
\usepackage{times} 
\fi

\ifdefined\ParSkip 
\usepackage{parskip} 
\fi

\newtheorem{theorem}{Theorem}

\newtheorem{corollary}{Corollary}
\newtheorem{proposition}{Proposition}
\theoremstyle{definition}

\theoremstyle{definition}

\makeatletter
\newcommand*\rel@kern[1]{\kern#1\dimexpr\macc@kerna}
\newcommand*\widebar[1]{%
  \begingroup
  \def\mathaccent##1##2{%
    \rel@kern{0.8}%
    \overline{\rel@kern{-0.8}\macc@nucleus\rel@kern{0.2}}%
    \rel@kern{-0.2}%
  }%
  \macc@depth\@ne
  \let\math@bgroup\@empty \let\math@egroup\macc@set@skewchar
  \mathsurround\z@ \frozen@everymath{\mathgroup\macc@group\relax}%
  \macc@set@skewchar\relax
  \let\mathaccentV\macc@nested@a
  \macc@nested@a\relax111{#1}%
  \endgroup
}
\makeatother

\DeclareMathOperator*{\argmin}{argmin}

\def\R{\mathbb{R}}

\def\E{\mathbb{E}}
\def\P{\mathbb{P}}

\def\cA{\mathcal{A}}

\def\cX{\mathcal{X}}
\def\cY{\mathcal{Y}}
\def\cZ{\mathcal{Z}}

\def\ind#1{\mathds{1}\left\{#1\right\}}

